\documentclass[12pt]{article}
\usepackage[utf8]{inputenc}
\usepackage[usenames,dvipsnames,svgnames,table]{xcolor}
\usepackage[margin=1in]{geometry}
\usepackage{graphicx}
\usepackage{subcaption}
\usepackage{tcolorbox}
\usepackage{mathtools}
\usepackage{times}
\usepackage{sp}
\usepackage[numbers,sort&compress]{natbib}
\usepackage{thm-restate}

\usepackage{pgf,tikz,pgfplots}
\pgfplotsset{compat=1.15}
\usepackage{mathrsfs}
\usetikzlibrary{arrows}

\usepackage{quiver}

\usepackage{subcaption}

\usepackage[colorlinks=true, citecolor=orange, linkcolor=black]{hyperref}

\newcommand{\VCG}{\texttt{VCG}}
\newcommand{\WT}{\texttt{WT}}

\newcommand{\eff}{\texttt{eff}}

\newcommand{\bTheta}{\boldsymbol{\Theta}}

\title{{\bf Revenue-Optimal Efficient Mechanism Design \\with General Type Spaces}}
\author{Siddharth Prasad\and Maria-Florina Balcan\and Tuomas Sandholm}
\date{}

\begin{document}

\maketitle

\begin{abstract}
  We derive the revenue-optimal efficient (welfare-maximizing) mechanism in a general multidimensional mechanism design setting when type spaces---that is, the underlying domains from which agents' values come from---can capture arbitrarily complex informational constraints about the agents. Type spaces can encode information about agents representing, for example, machine learning predictions of agent behavior, institutional knowledge about feasible market outcomes (such as item substitutability or complementarity in auctions), and correlations between multiple agents. Prior work has only dealt with {\em connected} type spaces, which are not expressive enough to capture many natural kinds of constraints such as disjunctive constraints. We provide two characterizations of the optimal mechanism based on allocations and connected components; both make use of an underlying network flow structure to the mechanism design. Our results significantly generalize and improve the prior state of the art in revenue-optimal efficient mechanism design. They also considerably expand the scope of what forms of agent information can be expressed and used to improve revenue. 
\end{abstract}

\section{Introduction}
Efficient mechanism design is the science of implementing outcomes that maximize economic value among strategic self-interested agents. It is the cornerstone of prominent real-world market design applications including combinatorial auctions for spectrum~\citep{cramton2013spectrum,leyton2017economics} and sourcing~\citep{Hohner03:Combinatorial,Sandholm06:Changing,Sandholm13:Very}, and Internet display advertisement auctions~\citep{Edelman07:Internet,Varian07:Position}. Additional modern applications include proposed redesigns of financial exchanges~\citep{budish2023flow}, better incentive-aware recommender ecosystems~\citep{prasad2023content,boutilier2024recommender}, and auctions for large language models~\citep{duetting2024mechanism,hajiaghayi2024ad}.

Our focus in the present paper is the design of {\em revenue-maximizing pricing rules} for efficient mechanism design. Our approach is to leverage information about the agents available to the mechanism designer through the agents' {\em type spaces}, which encode constraints on agents' private values (or {\em types}) that are known to the mechanism designer to hold before the private types are elicited. When type spaces are {\em connected},~\citet{krishna1998efficient} and~\citet{balcan2023bicriteria} show that the {\em weakest-type (WT) mechanism} is revenue optimal subject to efficiency, incentive compatibility (agents are best off reporting their true values to the mechanism designer), and individual rationality (no agent is charged more than their reported value for the chosen outcome). The WT mechanism is a type-space-dependent enhancement of the classical~{\em\citet{Vickrey61:Counterspeculation}-\citet{Clarke71:Multipart}-\citet{Groves73:Incentives} (VCG)} mechanism. Informally, while VCG charges an agent her negative externality on the other agents---as measured relative to her exiting the mechanism---WT charges an agent her negative externality relative to the weakest type consistent with her type space. WT always extracts at least as much revenue as VCG (\citet{balcan2023bicriteria} quantify this improvement based on how informative the type spaces are), so it can help overcome the well-documented issues of low revenue that VCG suffers from~\citep{Ausubel06:Lovely,ausubel2023vcg}.

However, connected type spaces are unable to express many natural constraints about agent types. For example, an auction designer might know that a bidder will bid for development rights in exactly one of two geographic regions, but not know which one. This kind of exclusivity constraint/disjunction can only be represented by a disconnected type space. Similarly, the auction designer might know that a bidder, if she submits a bid for a particular item, will bid at least $\$5$ million for that item. The type space here is also inherently disconnected since it allows for either no ($\$0$) bid or a bid exceeding $\$5$ million, but nothing in between. Another cause of disconnectedness is discrete type expression. For example, the FCC has experimented with ``click-box'' bidding to prevent collusive bidding via bid signaling. Here, bids are placed by clicking on the desired spectrum licenses; the bid value is given by fixed increments which precludes the ability to bid any dollar amount~\citep{Cramton00:Collusive,Bajari09:Auction}.

In this work we show that the prior state of the art---the WT mechanism---generates suboptimal revenue when type spaces are disconnected, and
derive the revenue-optimal efficient mechanism for general disconnected type spaces.

\subsection{Our Contributions}

We derive the revenue-optimal efficient mechanism for general agent type spaces. Prior work on efficient multidimensional mechanism design~\citep{krishna1998efficient,balcan2023bicriteria,balcan2025increasing} and efficient trade~\citep{Myerson83:Efficient,cramton1987dissolving} has only considered connected type spaces. In Section~\ref{sec:problem-formulation} we set up the general backdrop of multidimensional mechanism design, review the current state of knowledge on revenue-optimal efficient mechanism design for connected type spaces, and present examples of natural constraints on agent types that are disconnected and therefore outside the scope of prior work. In Section~\ref{sec:wt-suboptimality} we present a simple example showing that the vanilla WT mechanism is suboptimal for disconnected type spaces. In Section~\ref{sec:allocational} we derive the optimal efficient mechanism in terms of {\em allocation-wise Groves mechanisms}, a generalization of the classic~\citet{Groves73:Incentives} mechanism with a pricing scheme that depends more intimately on the efficient allocation. In Section~\ref{sec:component-wise} we provide an alternate characterization of the optimal efficient mechanism based on the decomposition of the type space into connected components, and {\em component-wise Groves mechanisms}. In Section~\ref{sec:example} we illustrate our approach with a simple example.

Key to both our characterizations (in terms of allocations and connected components) is an underlying network flow structure to the optimal efficient mechanism that we establish. Either one of these characterizations could be more useful than the other depending on how the mechanism designer's knowledge about the agents' possible types is represented or learned.

\subsection{Related Work}

\paragraph{Mechanism Design with Complex Type Spaces} There is a small body of work on mechanism design with type spaces that do not conform to the usual assumptions---typically convexity, such as in the seminal work of~\citet{Myerson81:Optimal}---made in economics. \citet{skreta2006mechanism} derives the optimal mechanism for a single-parameter setting where agents' type spaces can be arbitrary measurable subsets of the real line.~\citet{monteiro2009abstract} studies incentive compatibility for general multidimensional type spaces.~\citet{Lovejoy06:Optimal} analyzes various characterizations of optimal mechanisms when agents have finite type spaces, and~\citet{mu2008mechanism} provide characterizations of incentive compatibility in this setting. The mechanisms in these aforementioned works are {\em not} efficient. 

\paragraph{Weakest Types} \citet{Myerson83:Efficient} were the first to introduce the notion of a weakest or ``worst-off'' type in their seminal work on the impossibility of efficient budget-balanced trade; this idea was extended by~\citet{cramton1987dissolving}. The WT mechanism that we generalize and improve upon in the present paper was introduced by~\citet{krishna1998efficient} and generalized by~\citet{balcan2023bicriteria}. \citet{guo2013undominated} study a similar notion in the context of redistribution mechanisms.

\paragraph{Revenue-Maximizing Mechanism Design Without an Efficiency Constraint} \citet{Myerson81:Optimal} first characterized the revenue optimal auction for the sale of a single indivisible item to multiple buyers. Since then, progress on the general multi-item multi-bidder problem has been limited to very special cases ({\em e.g.,}~\citet{manelli2006bundling,Cai12:Optimal}). A popular modern approach to revenue-maximizing mechanism design is via machine learning. Sandholm and Likhodedov~\citep{Likhodedov04:Methods,sandholm2015automated} were the first to automatically optimize mechanism parameters from samples using gradient methods.~\citet{balcan2005mechanism,balcan2007mechanism} made the first connection between machine learning theory and mechanism design to theoretically assess the sample complexity of optimizing mechanisms from data. Data-driven mechanism design has since grown into a rich area spanning theory~\citep{balcan2025generalization} to modern practical methods like deep learning~\citep{dutting2019optimal,curry2023differentiable,wang2024gemnet} and diffusion models~\citep{wang2025bundleflow}. The design of pricing rules for {\em efficient} mechanism design has largely been left unexplored by the machine learning for mechanism design literature. The new characterization results we present here should serve as a timely launchpad for a new research strand along this vein.

\section{Problem Formulation, Mechanism Design Background, and Examples of Disconnected Type Spaces}\label{sec:problem-formulation}

In multidimensional mechanism design there is an abstract (finite) set $\Gamma$ of outcomes or {\em allocations}. There are $n$ agents, indexed by $i = 1,\ldots, n$, who have private values for each outcome. Agent $i$'s valuation function, or {\em type}, $v_i : \Gamma\to\R$, encodes the maximum value $v_i(\alpha)$ she is willing to pay for allocation $\alpha$ to be realized. We treat Agent $i$'s type as an element of $\R^{\Gamma}$. A full type profile across the agents is denoted by $\vec{v} = (v_1,\ldots, v_n)$, with $\vec{v}_{-i}$ denoting the joint type profile excluding $v_i$.

\paragraph{Type Spaces}
The {\em joint type space} of the agents is denoted by $\bTheta\subseteq\bigtimes_{i=1}^n\R^{\Gamma}$. While the mechanism designer does know the agents' private type profile $\vec{v}$ before types are elicited, he knows that $\vec{v}\in\bTheta$. Given revealed types $\vec{v}_{-i}$ of agents excluding $i$, Agent $i$'s induced type space is $\Theta_i(\vec{v}_{-i}) = \{v_i : (v_i,\vec{v}_{-i})\in\bTheta\}$. Let $\bTheta_{-i} = \{\vec{v}_{-i} : \exists v_i\in\R^{\Gamma}\text{ s.t. }(v_i,\vec{v}_{-i})\in\bTheta\}$. The mechanism designer's role is to elicit types $v_1,\ldots, v_n$ from the agents, choose an allocation $\alpha\in\Gamma$ based on the elicited types, and charge each agent a payment for the chosen allocation.

\paragraph{Efficiency, Incentive Compatibility, and Individual Rationality} The {\em efficient} allocation is $\alpha^{\eff}(\vec{v}) = \argmax_{\alpha\in\Gamma}\sum_{j=1}^n v_j(\alpha)$; a mechanism is {\em efficient} if it always implements $\alpha^{\eff}(\vec{v})$. We assume a fixed tie-breaking rule so that the argmin is unique. A pricing rule $\vec{p} = (p_1,\ldots, p_n)$ specifies a pricing function $p_i : \bTheta\to\R$ for each agent; $\vec{p}$ is {\em incentive compatible (IC)} if $v_i(\alpha^{\eff}(v_i,\vec{v}_{-i})) - p_i(v_i,\vec{v}_{-i})\ge v_i(\alpha^{\eff}(v_i',\vec{v}_{-i})) - p_i(v_i',\vec{v}_{-i})$ for all $(v_i,\vec{v}_{-i}), (v_i', \vec{v}_{-i})\in\bTheta$ and $\vec{p}$ is {\em individually rational (IR)} if $v_i(\alpha^{\eff}(v_i,\vec{v}_{-i})) - p_i(v_i,\vec{v}_{-i})\ge 0$ for all $(v_i,\vec{v}_{-i})\in\bTheta$. The revenue of an efficient mechanism with payment rule $\vec{p}$ on revealed type profile $\vec{v}$ is $\sum_{i=1}^n p_i(\vec{v})$.

\paragraph{Groves Mechanisms and Weakest Types}

A {\em Groves mechanism} is defined by functions $h_1,\ldots, h_n$ where $h_i : \bigtimes_{j\neq i}\R^{\Gamma}\to\R$ does not depend on Agent $i$'s revealed type $v_i$. It implements the efficient allocation $\alpha^{\eff}(\vec{v})$ via payments $p_i(\vec{v}) = h_i(\vec{v}_{-i}) - \sum_{j\neq i}v_j(\alpha^{\eff}(\vec{v}))$. All Groves mechanisms are efficient and IC. Two important specific Groves mechanisms are the classical{\em ~\citet{Vickrey61:Counterspeculation}-\citet{Clarke71:Multipart}-\citet{Groves73:Incentives} (VCG)} mechanism and the {\em weakest-type (WT)} mechanism~\citep{krishna1998efficient,balcan2023bicriteria} which are defined by the pricing rules $p_i^{\VCG}(\vec{v}) = w(0, \vec{v}_{-i}) - \sum_{j\neq i}v_j(\alpha^{\eff}(\vec{v}))$ and $p_i^{\WT}(\vec{v}) = \inf_{\widetilde{v}_i\in \Theta_i(\vec{v}_{-i})} w(\widetilde{v}_i,\vec{v}_{-i}) - \sum_{j\neq i}v_j(\alpha^{\eff}(\vec{v}))$, respectively. VCG and WT are both efficient, IC, and IR.

The VCG mechanism---despite its theoretical attractiveness---is nearly never used in practice due to its proclivity to generate unacceptably low revenue (or even zero revenue~\citep{ausubel2017market,ausubel2023vcg}). The WT mechanism is an attractive alternative since it preserves the theoretical virtues of VCG and can only yield improved revenue based on how much information about the agents is conveyed by $\bTheta$.

Our objective is to design better pricing rules that yield better revenues. As detailed in the subsequent section, when type spaces are connected---a stringent assumption on the structure of agent types (but nonetheless the predominant assumption in the mechanism design literature)---one cannot beat WT. 

\subsection{Optimal Efficient Mechanism Design with Connected Type Spaces}

We review the current state of knowledge on efficient mechanism design when type spaces are connected. The first two results are about the uniqueness of Groves mechanisms as the only efficient and IC mechanisms. The third is about the revenue optimality of the weakest type mechanism.

\begin{theorem}[Revenue Equivalence~\citep{Green77:Characterization,Holmstroem79:Groves}]\label{thm:uniqueness-of-prices}
    Suppose $\Theta_i(\vec{v}_{-i})$ is connected for every $\vec{v}_{-i}$. Let $\vec{p}$ be an IC pricing rule and let $\vec{p}'$ be any other pricing rule. Then, $\vec{p}'$ is IC if and only if there exist functions $h_i : \bTheta_{-i}\to\R$ such that $p_i'(\vec{v}) = p_i(\vec{v})+h_i(\vec{v}_{-i})$ for all $\vec{v}$. \footnote{This result holds more generally for any (not-necessarily-efficient) allocation function $f:\bTheta\to\Gamma$.}
\end{theorem}

\citet{nisan2007introduction} calls this result ``uniqueness of prices" and provides a self-contained proof.

\begin{theorem}[Uniqueness of Groves Mechanisms~\citep{Green77:Characterization,Holmstroem79:Groves}]\label{thm:groves-uniqueness}
    Suppose $\Theta_i(\vec{v}_{-i})$ is connected for every $\vec{v}_{-i}$ and let $\vec{p}$ be a pricing rule. Then, $\vec{p}$ is IC if and only if there exist functions $h_i:\bTheta_{-i}\to\R$ such that $p_i(\vec{v}) = h_i(\vec{v}_{-i}) - \sum_{j\neq i}v_j(\alpha^{\eff}(\vec{v}))$. In other words, the only efficient and IC mechanisms on connected type spaces are Groves mechanisms.
\end{theorem}

{\em Remark.} Theorems~\ref{thm:uniqueness-of-prices} and~\ref{thm:groves-uniqueness} are actually generalizations of the versions derived by~\citet{Holmstroem79:Groves} and presented in~\citet{nisan2007introduction} (the proofs are identical so we omit them). Those versions do not allow Agent $i$'s type space to vary based on the revealed types $\vec{v}_{-i}$ of the other agents; there is just a fixed type space $\Theta_i$ for each agent, and $\Theta_i$ needs to be connected. In contrast, we require that $\Theta_i(\vec{v}_{-i})$ is connected for each $\vec{v}_{-i}$. We present here a concrete example for which Theorems~\ref{thm:uniqueness-of-prices} and~\ref{thm:groves-uniqueness} apply but the original versions from~\citet{Holmstroem79:Groves} do not. Consider an auction of a single item where the auctioneer does not know the quality of the item but knows that all bids will be clustered around either a high value or a low value---say $\bTheta = \{\vec{v}\in\R^n : v_i\in [H-\varepsilon, H+\varepsilon]\;\forall\;i\}\cup\{\vec{v}\in\R^n : v_i\in [L-\varepsilon, L+\varepsilon]\;\forall\;i\}$. While $\bTheta$ is disconnected, $\Theta_i(\vec{v}_{-i})$ is connected for every $\vec{v}_{-i}\in\bTheta_{-i}$ since revealed types $\vec{v}_{-i}$ determine whether $v_i$ is a low bid or a high bid.

\begin{theorem}[Optimality of Weakest Type~\citep{krishna1998efficient,balcan2023bicriteria}]
    Suppose $\Theta_i(\vec{v}_{-i})$ is connected for every $\vec{v}_{-i}$. Let $\vec{p}$ be any IC and IR pricing rule. Then $p_i(\vec{v})\le p_i^{\WT}(\vec{v})$ for all $i$ and all $\vec{v}\in\bTheta$.
\end{theorem}

\subsection{Examples of Disconnected Type Spaces}

Here we present three examples of natural constraints on agent types for which $\Theta_i(\vec{v}_{-i})$ is a disconnected set, illustrating the need for a more general theory than the current one.

\begin{itemize}
    \item {\em Exclusivity constraints.} Constraints of the form ``Bidder $i$ will place a bid for development rights exceeding $\$10$ million in either San Francisco or New York City, but not both'' correspond to disconnected type spaces of the form $\Theta_i = \{v_i : v_i(\text{SF})\ge 10\text{ million}, v_i(\text{NYC})=0\}\cup\{v_i : v_i(\text{SF})=0, v_i(\text{NYC})\ge 10\text{ million}\}$.
    \item {\em Conditionals.} Constraints of the form ``if Bidder $i$ bids on the bundle of items $\{A,B\}$, she will bid at least $\$5$ million'' correspond to disconnected type spaces that looks like $\Theta_i = \{v_i : v_i(\{A,B\}) = 0\text{ OR } v_i(\{A,B\})\ge 5\text{ million}\}$.
\end{itemize}

The above two kinds of constraints are natural in multi-item auctions since the auction designer likely does not know the specific items/packages a bidder will bid on, but has more refined knowledge about the value of any (hypothetical) bid.

\begin{itemize}
    \item {\em Discrete types.} Type spaces of the form $\Theta_i = \{v_i : v_i(\alpha)\equiv 0\pmod{1000}, v_i(A)\ge 5000\}$ convey that values for allocation $\alpha$ are expressed in increments of $\$1000$, starting at $\$5000$. 
\end{itemize}

Discretized type expression might be a natural part of an auction interface. For example, the FCC experimented with ``click-box bidding'' to prevent collusion wherein bidders can signal to other bidders via the numerical value of their bids~\citep{Cramton00:Collusive,Bajari09:Auction}.

\section{Example Illustrating Sub-optimality of Vanilla Weakest Type}\label{sec:wt-suboptimality}

Consider the following example of a two-item auction where a bidder has a disconnected type space: there are two items $A, B$ for sale, three bidders submit XOR bids $v_1(A) = 5$, $v_2(B) = 3$, and $v_3(A) = 1$ (under the XOR bidding language~\citep{Sandholm02:Algorithm} each bidder can only win a package they explicitly bid for, which effectively means bidders value packages they did not bid for at zero). The type space for Bidder 1 is $$\Theta_1 = \left\{(v_1(A), 0, 0) : v_1(A)\ge 4\right\}\cup\left\{(0, v_1(B), 0):v_1(B) \ge 4\right\}$$ which says that Bidder $1$ wants either $A$ or $B$, but not both, and will place a bid of at least $\$4$ on her desired item (the third coordinate represents $v_1(AB)$, which is zero). This is an example of an exclusivity constraint discussed previously. $\Theta_1\subseteq\R^2$ (ignoring the third coordinate which is always zero) is a disjoint union of two rays with one on the $v_1(A)$ axis and one on the $v_1(B)$ axis. All other bidders' type spaces are unrestricted. 
Bidders $1$ and $2$ win items $A$ and $B$, respectively, in the efficient allocation. VCG charges Bidder $1$ $p^{\VCG}_1 = 4 - 3 = 1$. WT charges Bidder $1$ $p^{\WT}_1 = \min\{7, 5\} - 3 = 2$. A better IC and IR payment scheme for Bidder $1$ that is still efficient is: ``if Bidder $1$ wins either item she bid on, she pays $\$4$''. That payment scheme extracts a payment of $\$4$ from Bidder 1, showing that WT is suboptimal here. 

\section{Characterization of the Optimal Efficient Mechanism}

We now derive the optimal efficient mechanism. We provide two equivalent characterizations. The first (Section~\ref{sec:allocational}) is via a decomposition of the type space based on allocations. The second (Section~\ref{sec:component-wise}) is based on the decomposition of the type space into connected components. Key to both our characterizations is an underlying network flow structure. Either one of these characterizations could be more useful than the other depending on how the mechanism designer's knowledge about the agents' possible types is represented or learned. Section~\ref{sec:example} contains an illustrative example.

\subsection{Allocational Characterization of the Optimal Efficient Mechanism}\label{sec:allocational}

Let $\Theta_i^{\alpha}(\vec{v}_{-i}) = \left\{v_i \in \Theta_i(\vec{v}_{-i}) : \alpha^{\texttt{eff}}(v_i,\vec{v}_{-i}) = \alpha\right\}$ be the set of types $v_i$ for Agent $i$ leading to efficient allocation $\alpha$. These sets form a partition of Agent $i$'s type space: $\Theta_i(\vec{v}_{-i}) = \bigcup_{\alpha\in\Gamma}\Theta_i^{\alpha}(\vec{v}_{-i}).$

\paragraph{Allocation-wise Groves Mechanisms} We define a large class of pricing rules, not all of which are IC, that contains all IC pricing rules. These generalize the vanilla Groves mechanisms. An {\em allocation-wise Groves mechanism} is defined by functions $h_i^{\alpha} : \bTheta_{-i}\to\R$ for every allocation $\alpha\in\Gamma$ and every agent $i$. It charges Agent $i$ $$p_i(\vec{v}) = h_i^{\alpha^{\eff}(\vec{v})}(\vec{v}_{-i}) - \sum_{j\neq i}v_j(\alpha^{\eff}(\vec{v})).$$ So, while the term $h_i$ in a vanilla Groves mechanism cannot have any dependence on Agent $i$'s revealed type, the corresponding term in an allocation-wise Groves mechanism can depend on the efficient allocation induced by Agent $i$'s revealed type. Not all allocation-wise Groves mechanisms are IC, but, as the following lemma shows, it is a rich enough class to cover all IC mechanisms.

\begin{lemma}\label{lem:allocation}
    If $\vec{p}$ is IC, it is an allocation-wise Groves mechanism.
\end{lemma}

\begin{proof}
    For each $\alpha$, partition $\Theta_i^{\alpha}(\vec{v}_{-i})$ as a disjoint union of its connected components: $\Theta_i^{\alpha}(\vec{v}_{-i}) = \bigcup_{C\in\cC^{\alpha}}C$ ($\cC$ need not be finite). When restricted to any connected component $C\in\cC^{\alpha}$, $\vec{p}$ is a vanilla Groves mechanism (due to Theorem~\ref{thm:groves-uniqueness}). That is, there exists $h_i^C$ such that for all $v_i\in C$, $p_i(v_i,\vec{v}_{-i}) = h_i^C(\vec{v}_{-i}) - \sum_{j\neq i}v_j(\alpha)$. It is a standard fact that a pricing rule is IC if and only if it prescribes identical payments for any two types leading to the same allocation. That is, $p_i(v_i,\vec{v}_{-i}) = p_i(v_i', \vec{v}_{-i})$ for any $v_i, v_i'\in\Theta_i^{\alpha}(\vec{v}_{-i})$ ({\em e.g.,} Proposition 1.27 of~\citet{nisan2007introduction}). Hence the functions $h_i^C$ for each $C\in\cC^{\alpha}$ are all identical; let $h_i^{\alpha}$ be this function. Then $\vec{p}$ is the allocation-wise Groves mechanism given by the $h_i^{\alpha}$.
\end{proof}

We now characterize all IC and IR pricing rules that implement the efficient allocation.

\begin{theorem}\label{thm:allocation-IC-feasibility} A pricing rule $\vec{p}$ is IC and IR if and only if it is an allocation-wise Groves mechanism given, for each $i$, by $(h_i^{\alpha})_{\alpha\in\Gamma}$ that satisfies
\begin{equation}\label{eq:allocation-constraints}\tag{Constr.-$\Gamma$}\begin{aligned}
    &h_i^{\alpha}(\vec{v}_{-i})\le \inf_{\widetilde{v}_i\in\Theta_i^{\alpha}(\vec{v}_{-i})}w(\widetilde{v}_i,\vec{v}_{-i})\quad\forall\;\alpha\in\Gamma \\
    &h_i^{\alpha}(\vec{v}_{-i}) - h_i^{\beta}(\vec{v}_{-i})\le \inf_{\widetilde{v}_i\in\Theta_i^{\alpha}(\vec{v}_{-i})}w(\widetilde{v}_i, \vec{v}_{-i}) - \left[\widetilde{v}_i(\beta) + \sum_{j\neq i}v_j(\beta)\right]\quad\forall\;\alpha,\beta\in\Gamma.
\end{aligned}\end{equation}
\end{theorem}

\begin{proof}
    An allocation-wise Groves mechanism $(h_i^{\alpha})_{\alpha\in\Gamma}$ is IR if and only if 
    \begin{align*}
    v_i(\alpha^{\eff}(v_i,\vec{v}_{-i})) & - \left[h_i^{\alpha^{\eff}(v_i,\vec{v}_{-i})}(\vec{v}_{-i}) - \sum_{j\neq i}v_j(\alpha^{\eff}(v_i,\vec{v}_{-i}))\right]\ge 0\quad\forall\;(v_i,\vec{v}_{-i})\in\bTheta \\
    &\iff  h_i^{\alpha^{\eff}(v_i,\vec{v}_{-i})}(\vec{v}_{-i}) \le w(v_i,\vec{v}_{-i})\quad\forall\;(v_i,\vec{v}_{-i})\in\bTheta \\
    &\iff h_i^{\alpha}(\vec{v}_{-i}) \le w(v_i,\vec{v}_{-i})\quad\forall\; \alpha\in\Gamma, \vec{v}_{-i}\in\bTheta_{-i}, v_i\in\Theta_i^{\alpha}(\vec{v}_{-i}) \\
    & \iff h_i^{\alpha}(\vec{v}_{-i}) \le \inf_{\widetilde{v}_i\in\Theta_i^{\alpha}(\vec{v}_{-i})} w(\widetilde{v}_i,\vec{v}_{-i})\quad\forall\; \alpha\in\Gamma, \vec{v}_{-i}\in\bTheta_{-i}.
    \end{align*}
    It is IC if and only if for all $\vec{v}_{-i}\in\bTheta_{-i}$ and all $v_i, v_i'\in\Theta_i(\vec{v}_{-i})$, an agent of true type $v_i$ has no incentive to misreport $v_i'$ to the mechanism. Any allocation-wise Groves mechanism already satisfies these constraints for $v_i, v_i'\in\Theta_i^{\alpha}(\vec{v}_{-i})$, for every $\alpha\in\Gamma$, since when restricted to any $\Theta_i^{\alpha}$ it is equivalent to the vanilla Groves mechanism given by $h_i^{\alpha}$. So, for each $\vec{v}_{-i}\in\bTheta_{-i}$, it suffices to enforce IC constraints over all $v_i\in\Theta_i^{\alpha}(\vec{v}_{-i})$ and all $v_i'\in\Theta_i^{\beta}(\vec{v}_{-i})$, over every pair of differing allocations $\alpha,\beta\in\Gamma$. An allocation-wise Groves mechanism $(h_i^{\alpha})_{\alpha\in\Gamma}$ is therefore IC if and only if (for all $\vec{v}_{-i}\in\bTheta_{-i}$)
    \begin{align*}
        v_i &(\alpha) - \left[h_i^{\alpha}(\vec{v}_{-i}) - \sum_{j\neq i}v_j(\alpha)\right]\ge v_i(\beta) - \left[h_i^{\beta}(\vec{v}_{-i}) - \sum_{j\neq i}v_j(\beta)\right]\quad\forall\;\alpha,\beta\in\Gamma, v_i\in\Theta_i^{\alpha}(\vec{v}_{-i}) \\
        &\iff h_i^{\alpha}(\vec{v}_{-i}) - h_i^{\beta}(\vec{v}_{-i})\le w(v_i, \vec{v}_{-i}) - \left[v_i(\beta) + \sum_{j\neq i}v_j(\beta)\right]\quad\forall\;\alpha,\beta\in\Gamma, v_i\in\Theta_i^{\alpha}(\vec{v}_{-i}) \\ 
        &\iff h_i^{\alpha}(\vec{v}_{-i}) - h_i^{\beta}(\vec{v}_{-i})\le \inf_{\widetilde{v}_i\in\Theta_i^{\alpha}(\vec{v}_{-i})}w(\widetilde{v}_i, \vec{v}_{-i}) - \left[\widetilde{v}_i(\beta) + \sum_{j\neq i}v_j(\beta)\right]\quad\forall\;\alpha,\beta\in\Gamma.
    \end{align*}
    The theorem statement now follows from Lemma~\ref{lem:allocation}.
\end{proof}

Theorem~\ref{thm:allocation-IC-feasibility} seems to leave open the possibility that there is actually a Parento frontier of undominated revenue-maximial allocation-wise Groves mechanisms. However, that is not the case. It turns out that the revenue-optimal allocation-wise Groves mechanism is unique, which we prove next.

\begin{theorem}\label{thm:allocational-characterization}
The unique revenue-optimal mechanism subject to efficiency, IC, and IR is the allocation-wise Groves mechanism given by $(h_i^{\alpha})_{\alpha\in\Gamma}$ that maximizes $\sum_{\alpha\in\Gamma}h_i^{\alpha}$ subject to constraints~\eqref{eq:allocation-constraints}, for each $i$.
\end{theorem}

We will prove Theorem~\ref{thm:allocational-characterization} by interpreting the linear program \begin{equation}\label{eq:lp-allocation}\tag{LP-$\Gamma$}\max\left\{\sum_{\alpha\in\Gamma}h_i^{\alpha} :~\eqref{eq:allocation-constraints}\right\}\end{equation} from the perspective of network flow theory.

First, observe that the~\ref{eq:lp-allocation} is always feasible. Indeed, vanilla WT is always a feasible solution: let $h_i^{\alpha} = \inf_{\widetilde{v}_i\in\Theta_i(\vec{v}_{-i})}w(\widetilde{v}_i,\vec{v}_{-i})$ for all $\alpha\in\Gamma$. Constraints of the first form are clearly satisfied. Constraints of the second form are also clearly satisfied since the left-hand side is zero, and the right-hand side is always non-negative as $\alpha$ maximizes welfare for all $\widetilde{v}_i\in\Theta_i^{\alpha}(\vec{v}_{-i})$ (by definition of $\Theta_i^{\alpha}$). Vanilla VCG with $h_i^{\alpha} = w(0, \vec{v}_{-i})$ for all $\alpha\in\Gamma$ is also a feasible solution.

\begin{proof}[Proof of Theorem~\ref{thm:allocational-characterization}]

Consider the directed graph $G = (V,E)$ with vertices $V = \{s\}\cup\Gamma$ ($s$ is the source node), edges $E = (\{s\}\cup \Gamma)\times \Gamma$, and edge costs $$\begin{aligned} & \mathsf{cost}(s, \alpha) = \inf_{\widetilde{v}_i\in\Theta_i^{\alpha}(\vec{v}_{-i})}w(\widetilde{v}_i,\vec{v}_{-i})\quad\forall\;\alpha\in\Gamma \\ & \mathsf{cost}(\beta,\alpha) = \inf_{\widetilde{v}_i\in\Theta_i^{\alpha}(\vec{v}_{-i})}w(\widetilde{v}_i, \vec{v}_{-i}) - \left[\widetilde{v}_i(\beta) + \sum_{j\neq i}v_j(\beta)\right]\quad\forall\;\alpha,\beta\in\Gamma.\end{aligned}$$ In words, $G$ is the complete directed graph on vertex set $\Gamma$ with an additional source vertex $s$ and directed edges from $s$ to each vertex of $\Gamma$. Edges of the form $(s,\alpha)$ have cost equal to the welfare of the weakest type in $\Theta_i^{\alpha}(\vec{v}_{-i})$. Edges $(\beta,\alpha)$ have cost equal to the minimum welfare difference between allocations $\alpha$ and $\beta$ over all types in $\Theta_i^{\alpha}(\vec{v}_{-i})$.

 Linear program~\ref{eq:lp-allocation} precisely solves the \emph{single-source shortest paths} problem on $G$ with source node $s$.\footnote{More accurately, it is the dual of the minimum-cost flow LP.} That is, the optimal solution $(h_i^{\alpha})_{\alpha\in\Gamma}$ to~\ref{eq:lp-allocation} has the property that, for every $\alpha\in\Gamma$, $h_i^{\alpha}$ is the cost of the shortest (minimum-cost) $s\to\alpha$ path in $G$ (this is a standard fact from network flow theory; see, for example,~\citet{erickson2017linear}). A minimum cost $s\to\alpha$ path in $G$ can be equivalently computed as $\max\left\{h_i^{\alpha} :~\eqref{eq:allocation-constraints}\right\}$, showing that~\ref{eq:lp-allocation} yields the unique optimal efficient mechanism.
\end{proof}

The network flow interpretation here is similar in spirit to those used to understand IC mechanisms in~\citet{vohra2011mechanism}. Vohra's focus is on characterizing IC mechanisms in terms of the {\em existence} of bounded shortest paths (as witnessed by the no-negative-cycle condition) on graphs that are similar to ours. In contrast, our shortest path LPs are always feasible, and we use them to give a new insight into generalized Groves mechanisms. To our knowledge, this is the first time the network interpretation of incentive compatibility has been used to describe revenue-optimal efficient mechanisms.

\subsection{Connected-Component Characterization of the Optimal Efficient Mechanism}\label{sec:component-wise}

In this section we give an equivalent characterization of the revenue-optimal efficient, IC, and IR mechanism based on decompositions of the type space into its connected components. The characterization does not make explicit reference to the underlying space of allocations $\Gamma$.

Given $\vec{v}_{-i}$, decompose $\Theta_i(\vec{v}_{-i})$ into a disjoint union of its connected components $\cC(\vec{v}_{-i})$, that is, $\Theta_i(\vec{v}_{-i}) = \bigcup_{C\in\cC(\vec{v}_{-i})}C$. We assume in this section that $\cC(\vec{v}_{-i})$ is finite. Let $C(v_i,\vec{v}_{-i})\in\cC(\vec{v}_{-i})$ denote the connected component $v_i$ lies in.

\paragraph{Component-wise Groves Mechanisms} A {\em component-wise Groves mechanism} is defined by functions $h_i^C:\bTheta_{-i}\to\R$ for every connected component $C\in\cC(\vec{v}_{-i})$ and for every agent $i$. It charges Agent $i$ $$p_i(\vec{v}) = h_i^{C(v_i,\vec{v}_{-i})}(\vec{v}_{-i}) - \sum_{j\neq i}v_j(\alpha^{\eff}(\vec{v})).$$

\begin{lemma}\label{lem:component}
    If $\vec{p}$ is IC, it is a component-wise Groves mechanism.
\end{lemma}

\begin{proof}
    Let $\vec{p}$ be IC. When restricted to a connected component $C\in\cC(\vec{v}_{-i})$, $\vec{p}$ is a Groves mechanism due to Theorem~\ref{thm:groves-uniqueness}. That is, there exists $h_i^C$ such that for all $v_i\in C$, $p_i(v_i,\vec{v}_{-i}) = h_i^C(\vec{v}_{-i}) - \sum_{j\neq i}v_j(\alpha^{\eff}(\vec{v}))$. So $\vec{p}$ is a component-wise Groves mechanism given by $(h_i^C)_{C\in\cC(\vec{v}_{-i})}$.
\end{proof}

\begin{theorem}\label{thm:cc-IC-feasibility}
    A pricing rule $\vec{p}$ is IC and IR if and only if it is a component-wise Groves mechanism given, for each $i$ and each $\vec{v}_{-i}$, by $(h_i^C)_{C\in\cC(\vec{v}_{-i})}$ that satisfies 
    \begin{equation*}\label{eq:cc-constraints}\tag{Constr.-$\cC$}\begin{aligned}
    &h_i^{C}\le \inf_{\widetilde{v}_i^C\in C}w(\widetilde{v}_i^C,\vec{v}_{-i})\;\forall\;C\in\cC \\
    &h_i^{C}-h_i^{D}\le\inf_{\substack{\widetilde{v}_i^C\in C \\\widetilde{v}_i^D\in D}} w(\widetilde{v}_i^C, \vec{v}_{-i}) - \left[\widetilde{v}_i^C(\alpha^{\eff}(\widetilde{v}_i^D,\vec{v}_{-i}))+\sum_{j\neq i}v_j(\alpha^{\eff}(\widetilde{v}_i^D,\vec{v}_{-i}))\right]\;\forall\;C, D\in \cC.
\end{aligned}\end{equation*}
\end{theorem}

\begin{proof}
 Component-wise Groves mechanism $(h_i^{C})_{C\in\cC(\vec{v}_{-i})}$ is IR if and only if 
 \begin{align*}
    v_i(\alpha^{\eff}(v_i,\vec{v}_{-i})) & - \left[h_i^{C(v_i,\vec{v}_{-i})}(\vec{v}_{-i}) - \sum_{j\neq i}v_j(\alpha^{\eff}(v_i,\vec{v}_{-i}))\right]\ge 0\quad\forall\;(v_i,\vec{v}_{-i})\in\bTheta \\
    &\iff  h_i^{C(v_i,\vec{v}_{-i})}(\vec{v}_{-i}) \le w(v_i,\vec{v}_{-i})\quad\forall\;(v_i,\vec{v}_{-i})\in\bTheta \\
    &\iff h_i^{C}(\vec{v}_{-i}) \le w(v_i,\vec{v}_{-i})\quad\forall\; \vec{v}_{-i}\in\bTheta_{-i}, C\in\cC(\vec{v}_{-i}), v_i\in C  \\
    & \iff h_i^{C}(\vec{v}_{-i}) \le \inf_{\widetilde{v}_i\in C} w(\widetilde{v}_i,\vec{v}_{-i})\quad\forall\; \vec{v}_{-i}\in\bTheta_{-i},C\in\cC(\vec{v}_{-i}).
    \end{align*}
    It is IC if and only if for all $\vec{v}_{-i}\in\bTheta_{-i}$ and all $v_i, v_i'\in\Theta_i(\vec{v}_{-i})$, an agent of true type $v_i$ has no incentive to misreport $v_i'$ to the mechanism. Any component-wise Groves mechanism already satisfies these constraints for $v_i, v_i'\in C$, for every $\vec{v}_{-i}\in\bTheta_{-i}$ and every $C\in\cC(\vec{v}_{-i})$, since when restricted to any connected component $C\in\cC(\vec{v}_{-i})$ it is equivalent to the vanilla Groves mechanism given by $h_i^{C}$. So, for each $\vec{v}_{-i}\in\bTheta_{-i}$, it suffices to enforce IC constraints over all $v_i\in C$ and all $v_i'\in D$, over every pair of differing connected components $C,D\in\cC(\vec{v}_{-i})$. A component-wise Groves mechanism $(h_i^{C})_{C\in\cC(\vec{v}_{-i})}$ is therefore IC if and only if (for all $\vec{v}_{-i}\in\bTheta_{-i}$; let $\cC = \cC(\vec{v}_{-i})$, $h_i^C = h_i^C(\vec{v}_{-i})$, $h_i^D=h_i^D(\vec{v}_{-i})$ for brevity)
    \begin{align*}
        v_i (&\alpha^{\eff}(v_i,\vec{v}_{-i})) - \left[h_i^{C} - \sum_{j\neq i}v_j(\alpha^{\eff}(v_i,\vec{v}_{-i}))\right]\\ &\quad\ge v_i(\alpha^{\eff}(v_i',\vec{v}_{-i})) - \left[h_i^{D} - \sum_{j\neq i}v_j(\alpha^{\eff}(v_i',\vec{v}_{-i}))\right]\quad\forall\;C, D\in\cC, v_i\in C, v_i'\in D  \\
        &\iff h_i^{C} - h_i^{D}\le w(v_i, \vec{v}_{-i}) - \left[v_i(\alpha^{\eff}(v_i',\vec{v}_{-i})) + \sum_{j\neq i}v_j(\alpha^{\eff}(v_i',\vec{v}_{-i}))\right]\\ &\hspace{26.5em}\forall\;C, D\in\cC, v_i\in C, v_i'\in D  \\ 
        &\iff h_i^{C} - h_i^{D}\le \inf_{\substack{\widetilde{v}_i^C\in C \\\widetilde{v}_i^D\in D}} w(\widetilde{v}_i^C, \vec{v}_{-i}) - \left[\widetilde{v}_i^C(\alpha^{\eff}(\widetilde{v}_i^D,\vec{v}_{-i}))+\sum_{j\neq i}v_j(\alpha^{\eff}(\widetilde{v}_i^D,\vec{v}_{-i}))\right]\\&\hspace{26.5em}\forall\;C,D\in\cC.
    \end{align*}
    The theorem statement now follows from Lemma~\ref{lem:component}.
\end{proof}

\begin{theorem}
    The unique revenue-optimal mechanism subject to efficiency, IC, and IR is the component-wise Groves mechanism given by, for each agent $i$, $(h_i^C)_{C\in\cC_{\vec{v}_{-i}}}$ that maximizes $\sum_{C\in\cC(\vec{v}_{-i})} h_i^C$ subject to the constraints~\eqref{eq:cc-constraints}.
\end{theorem}

\begin{proof}
    The proof is similar to that of Theorem~\ref{thm:allocational-characterization}.
    The LP $\max\{\sum_{C\in\cC(\vec{v}_{-i})h_i^C}:\eqref{eq:cc-constraints}\}$ solves the single-source shortest paths problem on the directed graph $G = (V, E)$ with vertices $V = \{s\}\cup \cC$, edges $E = (\{s\}\cup\cC)\times\cC$, and edge costs $$\begin{aligned} & \mathsf{cost}(s, C) = \inf_{\widetilde{v}_i^C\in C}w(\widetilde{v}_i^C,\vec{v}_{-i})\;\forall\;C\in\cC \\ & \mathsf{cost}(D,C) = \inf_{\substack{\widetilde{v}_i^C\in C\\\widetilde{v}_i^D\in D}} w(\widetilde{v}_i^C, \vec{v}_{-i}) - \left[\widetilde{v}_i^C(\alpha^{\eff}(\widetilde{v}_i^D,\vec{v}_{-i}))+\sum_{j\neq i}v_j(\alpha^{\eff}(\widetilde{v}_i^D,\vec{v}_{-i}))\right]\;\forall\;C, D\in\cC\end{aligned}$$ in the sense that the optimal $h_i^C(\vec{v}_{-i})$ is the cost of the shortest $s\to C$ path in $G$. It follows that maximizing $\sum_C h_i^C$ yields the unique revenue optimal efficient mechanism.
\end{proof}

Depending on the mechanism design setting, the connected component graph can be significantly smaller than the allocation graph from the previous section. For example, $|\Gamma|$ is exponential in combinatorial auctions so the graph from Section~\ref{sec:allocational} is prohibitively large. If type spaces are represented as a union of $K$ connected components, the graph in the present section has $K+1$ vertices and $K(K+1)$ edges, regardless of how large $|\Gamma|$ might be. The key disclaimer here is that we have not pursued the algorithmic question of how to compute the edge weights of these graphs---that is an important next step for future research.

\subsection{Example}\label{sec:example}

Consider the same example from Section~\ref{sec:wt-suboptimality} with $v_2(B) = 3$ and $v_3(A) = 1$. Figure~\ref{fig:main_fig} displays the partition of Bidder 1's ambient type space $\R^2$ into three regions labeled $\emptyset$, $A$, and $B$, in which she wins nothing, item $A$, or item $B$ in the efficient allocation, respectively. Suppose now that Bidder 1's type space is $\Theta_1^A\cup\Theta_1^B$ as displayed in the first row of Figure~\ref{fig:main_fig}. Since $\Theta_1^A$ and $\Theta_1^B$ are each connected, the allocation-wise graph and the component-wise graph are identical. Vanilla WT prescribes $h_i^A = h_i^B = \min\{5,7\} = 5$, with a payment of $2$ if $v_i\in\Theta_1^A$ and a payment of $4$ if $v_i\in\Theta^B$. We have $\inf_{\widetilde{v}_i^A\in\Theta^A}w(\widetilde{v}_i^A,\vec{v}_{-i}) = 7$ and $\inf_{\widetilde{v}_i^B\in\Theta^B}w(\widetilde{v}_i^B,\vec{v}_{-i}) = 5$, attained by $\widetilde{v}_i^A = (4,1)$ and $\widetilde{v}_i^B = (1,4)$, respectively. We have $\mathsf{cost}(B,A)=3$ with $\widetilde{v}_1^A = (4,3)$ attaining the infimum and $\mathsf{cost}(A,B) = 1$ with $\widetilde{v}_1^B = (1,4)$ attaining the infimum. The revenue-optimal mechanism thus sets $h_1^A = 7$ and $h_1^B = 5$, extracting a payment of $4$ from Bidder 1 independent of which component her true type lies in. The second row of Figure~\ref{fig:main_fig} displays a similar situation with a larger $\Theta_1^A$. The edge costs of $(s,A)$, $(s,B)$, and $(A,B)$ remain the same, but now $\mathsf{cost}(B,A) = 1$ with $\widetilde{v}_1^A = (4,5)$ attaining the infimum. The revenue-optimal mechanism here sets $h_1^A = 6$ and $h_1^B = 5$, extracting a payment of $3$ from Bidder 1 if her true type is in $\Theta_1^A$ and $4$ if it is in $\Theta_1^B$. So, the ``less precise'' knowledge conveyed by the larger $\Theta_1^A$ results in lower payment extracted.

\begin{figure}[t]
    \centering
    \begin{subfigure}[b]{0.45\textwidth}
        \centering
        \begin{tikzpicture}[scale=.65,line cap=round,line join=round,>=triangle 45,x=1cm,y=1cm]
        \begin{axis}[
            x=1cm,y=1cm,
            axis lines=middle,
            ymajorgrids=false,
            xmajorgrids=false,
            xmin=0, xmax=6.5,
            ymin=0, ymax=6.5,
            xtick={0,...,6},
            ytick={0,...,6},
            xlabel={$v_1(A)$},
            ylabel={$v_1(B)$},
            every axis x label/.style={at={(ticklabel* cs:1.04)}, anchor=north,},
            every axis y label/.style={at={(ticklabel* cs:1.02)}, anchor=east,},
        ]
        \filldraw[line width=1pt,fill=black,fill opacity=0.1] (1,4) -- (3, 6) -- (0.5,6)--(0.5,4)--cycle;
        \draw[color=black] (1.3,5.3) node {\Large $\Theta_1^B$};
        \filldraw[line width=1pt,fill=black,fill opacity=0.1] (4,1) -- (6,1) -- (6,3) -- (4,3) -- cycle;
        \draw[color=black] (5,2) node {\Large $\Theta_1^A$};

        \draw [line width=1pt] (0,3)--(1,3);
        \draw [line width=1pt] (1,3)--(1,0);
        \draw [line width=1pt] (1,3)--(5,7);

        \draw[color=black] (.5,1.5) node {\Large $\emptyset$};
        \draw[color=black] (5.5,6.3) node {\Large $A$};
        \draw[color=black] (3.5,6.3) node {\Large $B$};
        \end{axis}
        \end{tikzpicture}
        \label{fig:A}
    \end{subfigure}
        \hfill
    \begin{subfigure}[b]{0.5\textwidth}
        \centering
        \[\begin{tikzcd}[sep=large, every label/.append style = {font = \normalsize}]
            && {\boxed{A}} \\
            {\boxed{s}} \\
            && {\boxed{B}}
            \arrow["1"{description}, curve={height=-10pt}, dashed, from=1-3, to=3-3]
            \arrow["7"{description}, from=2-1, to=1-3]
            \arrow["5"{description}, from=2-1, to=3-3]
            \arrow["3"{description}, curve={height=-10pt}, dashed, from=3-3, to=1-3]
        \end{tikzcd}\]
        \label{fig:B}
    \end{subfigure}
    \label{fig:TOP}
    \begin{subfigure}[b]{0.45\textwidth}
        \centering
        \begin{tikzpicture}[scale=.65,line cap=round,line join=round,>=triangle 45,x=1cm,y=1cm]
        \begin{axis}[
            x=1cm,y=1cm,
            axis lines=middle,
            ymajorgrids=false,
            xmajorgrids=false,
            xmin=0, xmax=6.5,
            ymin=0, ymax=6.5,
            xtick={0,...,6},
            ytick={0,...,6},
            xlabel={$v_1(A)$},
            ylabel={$v_1(B)$},
            every axis x label/.style={at={(ticklabel* cs:1.04)}, anchor=north,},
            every axis y label/.style={at={(ticklabel* cs:1.02)}, anchor=east,},
        ]
        \filldraw[line width=1pt,fill=black,fill opacity=0.1] (1,4) -- (3, 6) -- (0.5,6)--(0.5,4)--cycle;
        \draw[color=black] (1.3,5.3) node {\Large $\Theta_1^B$};
        \filldraw[line width=1pt,fill=black,fill opacity=0.1] (4,1) -- (6,1) -- (6,5) -- (4,5) -- cycle;
        \draw[color=black] (5,3) node {\Large $\Theta_1^A$};

        \draw [line width=1pt] (0,3)--(1,3);
        \draw [line width=1pt] (1,3)--(1,0);
        \draw [line width=1pt] (1,3)--(5,7);

        \draw[color=black] (.5,1.5) node {\Large $\emptyset$};
        \draw[color=black] (5.5,6.3) node {\Large $A$};
        \draw[color=black] (3.5,6.3) node {\Large $B$};
        \end{axis}
        \end{tikzpicture}
        \label{fig:C}
    \end{subfigure}
    \hfill
    \begin{subfigure}[b]{0.5\textwidth}
        \centering
        \[\begin{tikzcd}[sep=large, every label/.append style = {font = \normalsize}]
            && {\boxed{A}} \\
            {\boxed{s}} \\
            && {\boxed{B}}
            \arrow["1"{description}, curve={height=-10pt}, dashed, from=1-3, to=3-3]
            \arrow["7"{description}, dashed, from=2-1, to=1-3]
            \arrow["5"{description}, from=2-1, to=3-3]
            \arrow["1"{description}, curve={height=-10pt}, from=3-3, to=1-3]
        \end{tikzcd}\]
        \label{fig:D}
    \end{subfigure}
    \caption{Examples of a disconnected type space $\Theta_1 = \Theta_1^A\cup\Theta_1^B$ and the corresponding graph $G$ encoding the optimal efficient mechanism. The solid edges in $G$ make up the tree of shortest paths.}
    \label{fig:main_fig}
\end{figure}
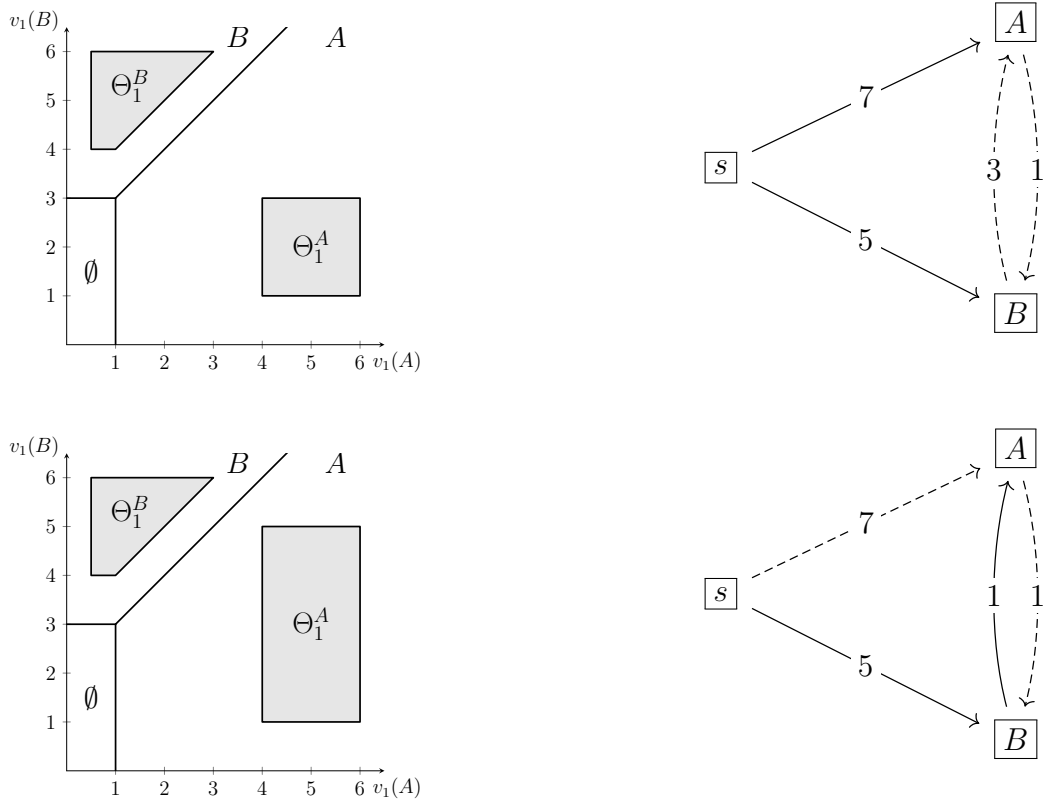

\section{Conclusions and Future Research}

We derived the revenue-optimal efficient mechanism when type spaces can be completely general. Our result significantly expands and generalizes the prior state of the art, the weakest type mechanism, that, while optimal for connected type spaces, is suboptimal for more general type spaces. Connected type spaces place a severe restriction on the kinds of knowledge structures that can be represented, and many natural informational constraints on agent types can only be described via a disconnected type space. We gave two characterizations of the optimal efficient mechanism, one via allocation-wise Groves mechanisms and one via component-wise Groves mechanisms. Both characterizations utilize the underlying network flow structure induced by incentive compatibility and individual rationality constraints.

Studying the computational aspects of our mechanisms is a pressing next research direction. Algorithms for computing the revenue-optimal efficient mechanism are an important next step to make our mechanisms practical. Such algorithms will be intrinsically tied to the {\em semantics} of agents' type spaces---how they are represented and learned will dictate how the optimal mechanism should be computed. 

An important complementary research direction is to develop techniques to {\em learn} representations of agent type spaces. Due to the generality of knowledge that type spaces can represent, we envision that modern machine learning models can be especially useful. In our setting, the learning problem is in a sense decoupled from the actual mechanism design. For example, prior learning-based approaches to mechanism design such as~\citet{dutting2019optimal,wang2024gemnet} learned the allocation function and payment function by representing them as neural networks---in contrast, in our setting one would use a machine learning model to learn as much information about the agents as possible, and then use our mechanism as a post-processor. Learning supports of (mixtures of) distributions~\citep{scott2005learning,dasgupta2005learning} and constraint learning~\citep{fajemisin2024optimization} are both relevant approaches for learning and representing type spaces from historical agent data. Extending the methodology here to other knowledge structures that incorporate distributional information and other forms of uncertainty is an interesting direction as well. 

\subsection*{Acknowledgements}
This material is based on work supported by the NSF under grants IIS-1901403, CCF-1733556, and RI-2312342, the ARO under award W911NF2210266, the Vannevar Bush Faculty Fellowship ONR N00014-23-1-2876, and NIH award A240108S001.

\newpage

\bibliography{references, dairefs}

\begin{thebibliography}{51}
\providecommand{\natexlab}[1]{#1}
\providecommand{\url}[1]{\texttt{#1}}
\expandafter\ifx\csname urlstyle\endcsname\relax
  \providecommand{\doi}[1]{doi: #1}\else
  \providecommand{\doi}{doi: \begingroup \urlstyle{rm}\Url}\fi

\bibitem[Ausubel and Baranov(2023)]{ausubel2023vcg}
Lawrence Ausubel and Oleg Baranov.
\newblock The {VCG} mechanism, the core, and assignment stages in auctions.
\newblock \emph{Working paper}, 2023.

\bibitem[Ausubel and Milgrom(2006)]{Ausubel06:Lovely}
Lawrence Ausubel and Paul Milgrom.
\newblock The lovely but lonely {V}ickrey auction.
\newblock In Peter Cramton, Yoav Shoham, and Richard Steinberg, editors, \emph{Combinatorial Auctions}, chapter~1. MIT Press, 2006.

\bibitem[Ausubel et~al.(2017)Ausubel, Aperjis, and Baranov]{ausubel2017market}
Lawrence Ausubel, Christina Aperjis, and Oleg Baranov.
\newblock Market design and the {FCC} incentive auction.
\newblock \emph{Presentation at the NBER Market Design Meeting}, 2017.

\bibitem[Bajari and Yeo(2009)]{Bajari09:Auction}
Patrick Bajari and Jungwon Yeo.
\newblock Auction design and tacit collusion in {FCC} spectrum auctions.
\newblock \emph{Information Economics and Policy}, 21\penalty0 (2):\penalty0 90--100, 2009.

\bibitem[Balcan and Blum(2007)]{balcan2007mechanism}
Maria-Florina Balcan and Avrim Blum.
\newblock Mechanism design, machine learning, and pricing problems.
\newblock \emph{ACM SIGecom Exchanges}, 7\penalty0 (1):\penalty0 34--36, 2007.

\bibitem[Balcan et~al.(2005)Balcan, Blum, Hartline, and Mansour]{balcan2005mechanism}
Maria-Florina Balcan, Avrim Blum, Jason~D Hartline, and Yishay Mansour.
\newblock Mechanism design via machine learning.
\newblock In \emph{IEEE Symposium on Foundations of Computer Science (FOCS)}. IEEE, 2005.

\bibitem[Balcan et~al.(2023)Balcan, Prasad, and Sandholm]{balcan2023bicriteria}
Maria-Florina Balcan, Siddharth Prasad, and Tuomas Sandholm.
\newblock Bicriteria multidimensional mechanism design with side information.
\newblock In \emph{Conference on Neural Information Processing Systems (NeurIPS)}, 2023.

\bibitem[Balcan et~al.(2025{\natexlab{a}})Balcan, Prasad, and Sandholm]{balcan2025increasing}
Maria-Florina Balcan, Siddharth Prasad, and Tuomas Sandholm.
\newblock Increasing revenue in efficient combinatorial auctions by learning to generate artificial competition.
\newblock In \emph{Proceedings of the AAAI Conference on Artificial Intelligence (AAAI)}, volume~39, pages 13572--13580, 2025{\natexlab{a}}.

\bibitem[Balcan et~al.(2025{\natexlab{b}})Balcan, Sandholm, and Vitercik]{balcan2025generalization}
Maria-Florina Balcan, Tuomas Sandholm, and Ellen Vitercik.
\newblock Generalization guarantees for multi-item profit maximization: Pricing, auctions, and randomized mechanisms.
\newblock \emph{Operations Research}, 73\penalty0 (2):\penalty0 648--663, 2025{\natexlab{b}}.

\bibitem[Boutilier et~al.(2024)Boutilier, Mladenov, and Tennenholtz]{boutilier2024recommender}
Craig Boutilier, Martin Mladenov, and Guy Tennenholtz.
\newblock Recommender ecosystems: A mechanism design perspective on holistic modeling and optimization.
\newblock In \emph{Proceedings of the AAAI Conference on Artificial Intelligence (AAAI)}, volume~38, pages 22575--22583, 2024.

\bibitem[Budish et~al.(2023)Budish, Cramton, Kyle, Lee, and Malec]{budish2023flow}
Eric Budish, Peter Cramton, Albert~S Kyle, Jeongmin Lee, and David Malec.
\newblock Flow trading.
\newblock Technical report, National Bureau of Economic Research, 2023.

\bibitem[Cai et~al.(2012)Cai, Daskalakis, and Weinberg]{Cai12:Optimal}
Yang Cai, Constantinos Daskalakis, and Matt Weinberg.
\newblock Optimal multi-dimensional mechanism design: Reducing revenue to welfare maximization.
\newblock In \emph{Proceedings of the Annual Symposium on Foundations of Computer Science (FOCS)}, 2012.

\bibitem[Clarke(1971)]{Clarke71:Multipart}
Ed~H. Clarke.
\newblock Multipart pricing of public goods.
\newblock \emph{Public Choice}, 1971.

\bibitem[Cramton(2013)]{cramton2013spectrum}
Peter Cramton.
\newblock Spectrum auction design.
\newblock \emph{Review of industrial organization}, 42:\penalty0 161--190, 2013.

\bibitem[Cramton and Schwartz(2000)]{Cramton00:Collusive}
Peter Cramton and Jesse Schwartz.
\newblock Collusive bidding: Lessons from the {FCC} spectrum auctions.
\newblock \emph{Journal of Regulatory Economics}, 17:\penalty0 229--252, 2000.

\bibitem[Cramton et~al.(1987)Cramton, Gibbons, and Klemperer]{cramton1987dissolving}
Peter Cramton, Robert Gibbons, and Paul Klemperer.
\newblock Dissolving a partnership efficiently.
\newblock \emph{Econometrica: Journal of the Econometric Society}, pages 615--632, 1987.

\bibitem[Curry et~al.(2023)Curry, Sandholm, and Dickerson]{curry2023differentiable}
Michael Curry, Tuomas Sandholm, and John Dickerson.
\newblock Differentiable economics for randomized affine maximizer auctions.
\newblock In \emph{Proceedings of the Thirty-Second International Joint Conference on Artificial Intelligence (IJCAI)}, pages 2633--2641, 2023.

\bibitem[Dasgupta et~al.(2005)Dasgupta, Hopcroft, Kleinberg, and Sandler]{dasgupta2005learning}
Anirban Dasgupta, John Hopcroft, Jon Kleinberg, and Mark Sandler.
\newblock On learning mixtures of heavy-tailed distributions.
\newblock In \emph{46th Annual IEEE Symposium on Foundations of Computer Science (FOCS)}, pages 491--500. IEEE, 2005.

\bibitem[D{\"u}tting et~al.(2019)D{\"u}tting, Feng, Narasimhan, Parkes, and Ravindranath]{dutting2019optimal}
Paul D{\"u}tting, Zhe Feng, Harikrishna Narasimhan, David Parkes, and Sai~Srivatsa Ravindranath.
\newblock Optimal auctions through deep learning.
\newblock In \emph{International Conference on Machine Learning (ICML)}, pages 1706--1715. PMLR, 2019.

\bibitem[D\"{u}tting et~al.(2024)D\"{u}tting, Mirrokni, Paes~Leme, Xu, and Zuo]{duetting2024mechanism}
Paul D\"{u}tting, Vahab Mirrokni, Renato Paes~Leme, Haifeng Xu, and Song Zuo.
\newblock Mechanism design for large language models.
\newblock In \emph{Proceedings of the ACM Web Conference (WWW)}, pages 144--155, 2024.

\bibitem[Edelman et~al.(2007)Edelman, Ostrovsky, and Schwarz]{Edelman07:Internet}
Benjamin Edelman, Michael Ostrovsky, and Michael Schwarz.
\newblock Internet advertising and the generalized second-price auction: Selling billions of dollars worth of keywords.
\newblock \emph{The American Economic Review}, 97\penalty0 (1):\penalty0 242--259, March 2007.
\newblock ISSN 0002-8282.

\bibitem[Erickson(2017)]{erickson2017linear}
Jeff Erickson.
\newblock Linear programming.
\newblock \url{https://courses.grainger.illinois.edu/cs473/sp2017/notes/H-lp.pdf}, 2017.

\bibitem[Fajemisin et~al.(2024)Fajemisin, Maragno, and den Hertog]{fajemisin2024optimization}
Adejuyigbe~O Fajemisin, Donato Maragno, and Dick den Hertog.
\newblock Optimization with constraint learning: A framework and survey.
\newblock \emph{European Journal of Operational Research}, 314\penalty0 (1):\penalty0 1--14, 2024.

\bibitem[Green and Laffont(1977)]{Green77:Characterization}
J~Green and J-J Laffont.
\newblock Characterization of satisfactory mechanisms for the revelation of preferences for public goods.
\newblock \emph{Econometrica}, 45:\penalty0 427--438, 1977.

\bibitem[Groves(1973)]{Groves73:Incentives}
Theodore Groves.
\newblock Incentives in teams.
\newblock \emph{Econometrica}, 1973.

\bibitem[Guo et~al.(2013)Guo, Markakis, Apt, and Conitzer]{guo2013undominated}
Mingyu Guo, Evangelos Markakis, Krzysztof~R Apt, and Vincent Conitzer.
\newblock Undominated {G}roves mechanisms.
\newblock \emph{Journal of Artificial Intelligence Research}, 46:\penalty0 129--163, 2013.

\bibitem[Hajiaghayi et~al.(2024)Hajiaghayi, Lahaie, Rezaei, and Shin]{hajiaghayi2024ad}
MohammadTaghi Hajiaghayi, Sebastien Lahaie, Keivan Rezaei, and Suho Shin.
\newblock Ad auctions for {LLM}s via retrieval augmented generation.
\newblock In \emph{The Thirty-eighth Annual Conference on Neural Information Processing Systems (NeurIPS)}, 2024.

\bibitem[Hohner et~al.(2003)Hohner, Rich, Ng, Reid, Davenport, Kalagnanam, Lee, and An]{Hohner03:Combinatorial}
Gail Hohner, John Rich, Ed~Ng, Grant Reid, Andrew~J. Davenport, Jayant~R. Kalagnanam, Ho~Soo Lee, and Chae An.
\newblock Combinatorial and quantity-discount procurement auctions benefit {M}ars, {I}ncorporated and its suppliers.
\newblock \emph{Interfaces}, 33\penalty0 (1):\penalty0 23--35, 2003.

\bibitem[Holmstr{\"{o}}m(1979)]{Holmstroem79:Groves}
Bengt Holmstr{\"{o}}m.
\newblock Groves' scheme on restricted domains.
\newblock \emph{Econometrica}, 47\penalty0 (5):\penalty0 1137--1144, 1979.

\bibitem[Krishna and Perry(1998)]{krishna1998efficient}
Vijay Krishna and Motty Perry.
\newblock Efficient mechanism design.
\newblock \emph{Available at SSRN 64934}, 1998.

\bibitem[Leyton-Brown et~al.(2017)Leyton-Brown, Milgrom, and Segal]{leyton2017economics}
Kevin Leyton-Brown, Paul Milgrom, and Ilya Segal.
\newblock Economics and computer science of a radio spectrum reallocation.
\newblock \emph{Proceedings of the National Academy of Sciences}, 114\penalty0 (28):\penalty0 7202--7209, 2017.

\bibitem[Likhodedov and Sandholm(2004)]{Likhodedov04:Methods}
Anton Likhodedov and Tuomas Sandholm.
\newblock Methods for boosting revenue in combinatorial auctions.
\newblock In \emph{Proceedings of the National Conference on Artificial Intelligence (AAAI)}, pages 232--237, San Jose, CA, 2004.

\bibitem[Lovejoy(2006)]{Lovejoy06:Optimal}
William~S. Lovejoy.
\newblock Optimal mechanisms with finite agent types.
\newblock \emph{Management Science}, 53\penalty0 (5):\penalty0 788--803, 2006.

\bibitem[Manelli and Vincent(2006)]{manelli2006bundling}
Alejandro~M Manelli and Daniel~R Vincent.
\newblock Bundling as an optimal selling mechanism for a multiple-good monopolist.
\newblock \emph{Journal of Economic Theory}, 127\penalty0 (1):\penalty0 1--35, 2006.

\bibitem[Monteiro(2009)]{monteiro2009abstract}
Paulo~Klinger Monteiro.
\newblock Abstract types and distributions in independent private value auctions.
\newblock \emph{Economic Theory}, 40\penalty0 (3):\penalty0 497--507, 2009.

\bibitem[Mu'alem and Schapira(2008)]{mu2008mechanism}
Ahuva Mu'alem and Michael Schapira.
\newblock Mechanism design over discrete domains.
\newblock In \emph{Proceedings of the 9th ACM Conference on Electronic Commerce (EC)}, pages 31--37, 2008.

\bibitem[Myerson(1981)]{Myerson81:Optimal}
Roger Myerson.
\newblock Optimal auction design.
\newblock \emph{Mathematics of Operation Research}, 6:\penalty0 58--73, 1981.

\bibitem[Myerson and Satterthwaite(1983)]{Myerson83:Efficient}
Roger Myerson and Mark Satterthwaite.
\newblock Efficient mechanisms for bilateral trading.
\newblock \emph{Journal of Economic Theory}, 28:\penalty0 265--281, 1983.

\bibitem[Nisan(2007)]{nisan2007introduction}
Noam Nisan.
\newblock Introduction to mechanism design (for computer scientists).
\newblock \emph{Algorithmic Game Theory}, 9:\penalty0 209--242, 2007.

\bibitem[Prasad et~al.(2023)Prasad, Mladenov, and Boutilier]{prasad2023content}
Siddharth Prasad, Martin Mladenov, and Craig Boutilier.
\newblock Content prompting: Modeling content provider dynamics to improve user welfare in recommender ecosystems.
\newblock \emph{arXiv preprint arXiv:2309.00940}, 2023.

\bibitem[Sandholm(2002)]{Sandholm02:Algorithm}
Tuomas Sandholm.
\newblock Algorithm for optimal winner determination in combinatorial auctions.
\newblock \emph{Artificial Intelligence}, 135:\penalty0 1--54, January 2002.

\bibitem[Sandholm(2013)]{Sandholm13:Very}
Tuomas Sandholm.
\newblock Very-large-scale generalized combinatorial multi-attribute auctions: Lessons from conducting \$60 billion of sourcing.
\newblock In Zvika Neeman, Alvin Roth, and Nir Vulkan, editors, \emph{Handbook of Market Design}. Oxford University Press, 2013.

\bibitem[Sandholm and Likhodedov(2015)]{sandholm2015automated}
Tuomas Sandholm and Anton Likhodedov.
\newblock Automated design of revenue-maximizing combinatorial auctions.
\newblock \emph{Operations Research}, 63\penalty0 (5):\penalty0 1000--1025, 2015.

\bibitem[Sandholm et~al.(2006)Sandholm, Levine, Concordia, Martyn, Hughes, Jacobs, and Begg]{Sandholm06:Changing}
Tuomas Sandholm, David Levine, Michael Concordia, Paul Martyn, Rick Hughes, Jim Jacobs, and Dennis Begg.
\newblock Changing the game in strategic sourcing at {P}rocter \& {G}amble: Expressive competition enabled by optimization.
\newblock \emph{Interfaces}, 36\penalty0 (1):\penalty0 55--68, 2006.

\bibitem[Scott and Nowak(2005)]{scott2005learning}
Clayton Scott and Robert Nowak.
\newblock Learning minimum volume sets.
\newblock \emph{Advances in Neural Information Processing Systems (NeurIPS)}, 18, 2005.

\bibitem[Skreta(2006)]{skreta2006mechanism}
Vasiliki Skreta.
\newblock Mechanism design for arbitrary type spaces.
\newblock \emph{Economics Letters}, 91\penalty0 (2):\penalty0 293--299, 2006.

\bibitem[Varian(2007)]{Varian07:Position}
Hal~R. Varian.
\newblock Position auctions.
\newblock \emph{International Journal of Industrial Organization}, pages 1163--1178, 2007.

\bibitem[Vickrey(1961)]{Vickrey61:Counterspeculation}
William Vickrey.
\newblock Counterspeculation, auctions, and competitive sealed tenders.
\newblock \emph{Journal of Finance}, 1961.

\bibitem[Vohra(2011)]{vohra2011mechanism}
Rakesh~V Vohra.
\newblock \emph{Mechanism design: a linear programming approach}, volume~47.
\newblock Cambridge University Press, 2011.

\bibitem[Wang et~al.(2024)Wang, Jiang, and Parkes]{wang2024gemnet}
Tonghan Wang, Yanchen Jiang, and David~C Parkes.
\newblock Gemnet: Menu-based, strategy-proof multi-bidder auctions through deep learning.
\newblock In \emph{Proceedings of the 25th ACM Conference on Economics and Computation (EC)}, pages 1100--1100, 2024.

\bibitem[Wang et~al.(2025)Wang, Jiang, and Parkes]{wang2025bundleflow}
Tonghan Wang, Yanchen Jiang, and David~C Parkes.
\newblock Bundleflow: Deep menus for combinatorial auctions by diffusion-based optimization.
\newblock \emph{arXiv preprint arXiv:2502.15283}, 2025.

\end{thebibliography}
\bibliographystyle{plainnat}

\end{document}